\documentclass[final]{dmtcs-episciences}

  \usepackage{tikz}
  \usepackage{enumerate}
  \usepackage{amssymb,amsmath,amsthm}
  \usepackage[english]{babel}

\theoremstyle{plain}
\newtheorem{DE}{Definition}%[section]
\newtheorem{defeng}[DE]{Definition}
\newtheorem{theorem}[DE]{Theorem}
\newtheorem{lemma}[DE]{Lemma}

\newtheorem{remark}[DE]{Remark}

\newcommand{\cC}{\mathcal{C}}
\newcommand{\cF}{\mathcal{F}}
\newcommand{\cO}{\mathcal{O}}
\newcommand{\es}{\varnothing}
\newcommand{\IN}{\mathbb{N}}
\newcommand{\sm}{\setminus}
\newcommand{\preq}{\preccurlyeq}
\newcommand{\rank}{\operatorname{rank}}
\newcommand{\cutrk}[1]{\operatorname{cutrk}_{#1}}
\newcommand{\width}{\operatorname{width}}
\newcommand{\rw}{\operatorname{rw}}

\begin{document}

\title{On rank-width of (diamond, even hole)-free graphs}

\author{%
     Isolde Adler\affiliationmark{1}
\and Ngoc Khang Le\affiliationmark{2,3}%
     \thanks{Partially supported by ANR project Stint
       under reference ANR-13-BS02-0007 and by the LABEX MILYON
       (ANR-10-LABX-0070) of Universit\'e de Lyon, within the program
       ``Investissements d'Avenir'' (ANR-11-IDEX-0007) operated by the
       French National Research Agency (ANR).
     } 
\and Haiko M\"{u}ller\affiliationmark{1} \\
\and Marko Radovanovi\'c\affiliationmark{4}%
     \thanks{Partially supported by Serbian Ministry of Education, Science
       and Technological Development project 174033.
     }
\and Nicolas Trotignon\affiliationmark{2}
\and Kristina Vu\v{s}kovi\'c\affiliationmark{1,5}%
     \thanks{Partially supported by EPSRC grant EP/N019660/1,
       by Serbian Ministry of Education, Science and
       Technological Development projects 174033 and III44006.
     }
}

\affiliation{
  School of Computing, University of Leeds, United Kingdom \\
  CNRS, ENS de Lyon, LIP, France \\
  Universit\'e de Lyon, France \\
  Faculty of Mathematics, University of Belgrade, Serbia \\
  Faculty of Computer Science (RAF), Union University, Serbia
}


\received{2016-12-5}
\revised{2017-8-1}
\accepted{2017-8-3}
\publicationdetails{19}{2017}{1}{24}{2575}

\maketitle

\begin{abstract}
  We present a class of (diamond, even hole)-free graphs with no
  clique cutset that has unbounded rank-width.
        
  In general, even-hole-free graphs have unbounded rank-width, because chordal
  graphs are even-hole-free. A.\,A.~da Silva, A.~Silva and C.~Linhares-Sales
  (2010) showed that planar even-hole-free graphs have bounded rank-width, and
  N.\,K.~Le (2016) showed that even-hole-free graphs with no star cutset have
  bounded rank-width. A natural question is to ask, whether even-hole-free
  graphs with no clique cutsets have bounded rank-width. Our result gives a
  negative answer. Hence we cannot apply the meta-theorem by Courcelle,
  Makowsky and Rotics, which would provide efficient algorithms for a large
  number of problems, including the maximum independent set problem, whose
  complexity remains open for (diamond, even hole)-free graphs.
\end{abstract}

\section{Introduction}

In a graph $G$ a \emph{hole} is a chordless cycle of length at least four. A
hole is \emph{even} or \emph{odd} depending on the parity of the size of its
vertex set. An \emph{$n$-hole} is a hole on $n$ vertices. A graph $G$
\emph{contains} a graph $F$, if $F$ is isomorphic to an induced subgraph of
$G$. $G$ is \emph{$F$-free} if it does not contain $F$, and for a family of
graphs $\cF$, $G$ is \emph{$\cF$-free} if for every $F\in \cF$, $G$ does not
contain $F$. A \emph{diamond} is the graph obtained by removing one edge from
a complete graph on four vertices. In this paper we study (diamond, even
hole)-free graphs.

Even-hole-free graphs have been studied considerably in the last two decades
(see surveys \cite{kv-survey1, kv-survey2}), and yet some of the key
algorithmic questions remain open for this class. Finding a largest (weighted)
clique in an even-hole-free graph can be done in polynomial time. As observed
by Farber \cite{farber}, 4-hole-free graphs have $\cO(n^2)$ maximal cliques,
so one can list them all in polynomial time. One can do better for
even-hole-free graphs, by exploiting structural properties of the class. In
\cite{daSV} it is shown that every even-hole-free graph has a vertex whose
neighbourhood is \emph{chordal} (i.e.\ hole-free), and in \cite{actv} it is
shown that an ordering of such vertices can be found using LexBFS, resulting
in an $\cO(nm)$-time algorithm for maximum weighted clique problem for
even-hole-free graphs. This algorithm is in fact robust: for any input graph
$G$, it either outputs a maximum weighed clique of $G$ or a certificate that
$G$ is not even-hole-free. Even-hole-free graphs can be recognized in
polynomial time, as first shown in \cite{cckv-ehfrecognition}, with currently
best complexity of $\cO(n^{11})$ \cite{cl}. This result is based on a
decomposition theorem for even-hole-free graphs from \cite{daSVehfdecomp} that
states that every even-hole-free graph is either simple in some sense, or has
a star cutset or a 2-join. In \cite{achrs} it is shown that every
even-hole-free graph $G$ has a vertex whose neighborhood is a union of two
(possibly empty) cliques, implying that $\chi(G) \le 2\omega(G)-1$. Despite
all these attempts to understand the structure of even-hole-free graphs, the
complexity of the stable set and coloring problems remains open for this
class.
 
For several subclasses of even-hole-free graphs these problems are solved in
polynomial time. Of particular interest is the class of (diamond, even
hole)-free graphs. The class was first studied in \cite{kmv} where it was
shown that (diamond, even hole)-free graphs can be decomposed by bisimplicial
cutsets (a special type of a star cutset that consists of two, possibly empty,
cliques) and 2-joins. One of the consequences of this decomposition theorem is
the existence of a vertex that is either of degree 2 or is simplicial (i.e.,
its neighborhood is a clique), implying that the class is $\beta$-perfect, and
for every graph $G$ in the class $\chi(G) \le \omega(G)+1$. The
$\beta$-perfection implies that the class can be colored in polynomial time by
coloring greedily on a particular, easily constructible, ordering of vertices.
The complexity of the stable set problem remains open for this class.

One of the motivations for the study of even-hole-free graphs is their
connection to $\beta$-perfect graphs introduced by Markossian, Gasparian and
Reed \cite{mgr}. For a graph $G$, let $\delta(G)$ denote the minimum degree
of a vertex of $G$. Consider the following total order on $V(G)$: order the
vertices by repeatedly removing a vertex of minimum degree in the subgraph of
vertices not yet chosen and placing it after all the remaining vertices but
before all the vertices already removed. Coloring greedily on this order
gives the upper bound: $\chi(G) \le \beta(G)$, where $\beta(G) =
\max\{\delta(H)+1 : H \text{~is an induced subgraph of~} G\}$. A graph is
\emph{$\beta$-perfect} if for each induced subgraph $H$ of $G$,
$\chi(H)=\beta(H)$. It is easy to see that $\beta$-perfect graphs are a
proper subclass of even-hole-free graphs.

\emph{Tree-width} is a well-known graph invariant, introduced by Robertson and
Seymour in~\cite{RS-GM02}. Many problems that are NP-hard in general become
tractable on graph classes of bounded tree-width~\cite{Courcelle90}.
Similarly, \emph{clique-width}, introduced by Courcelle, Engelfriet and
Rozenberg in~\cite{CER-93}, allows for many hard problems to become tractable
on graph classes of bounded clique-width~\cite{CMR00}. This includes finding
the largest clique or independent set, and deciding if a colouring with at
most $k$ colors exists (for fixed $k\in \mathbb N$). While bounded tree-width
implies bounded clique-width, the converse is not true in general. Graph
classes of bounded tree-width are necessarily sparse. In contrast, there
exist dense graph classes with bounded clique-width. This makes clique-width
particularly interesting in the study of algorithmic properties of hereditary
graph classes. The notion of \emph{rank-width} was defined by Oum and Seymour
in~\cite{OS-rw}, where they use it for an approximation algorithm for
clique-width. They also show that rank-width and clique-width are equivalent,
in the sense that a graph class has bounded rank-width if, and only if, it has
bounded clique-width. Meanwhile, the structure of graphs of bounded
rank-width is studied widely, and it turns out that rank-width is an elegant
notion, that also provides a better understanding of graph classes of bounded
clique-width.

Rank-width of subclasses of even-hole-free graphs has also been studied. In
\cite{dsl} it is shown that planar even-hole-free graphs have tree-width at
most 49. In \cite{kl} it is shown that even-hole-free graphs with no star
cutset have bounded rank-width. Even-hole-free graphs in general do not have
bounded tree-, clique-, rank-width, as they contain all chordal graphs.
Algorithms for chordal graphs follow from their decomposition by clique
cutsets, and clique cutsets in general agree well with a number of problems,
including stable set and coloring. An example of even-hole-free graphs with no
clique cutset and unbounded rank-width is given in \cite{kl}, which is a
slight modification of the class of permutation graphs introduced in
\cite{gr}. In \cite{tk} Kloks claims a proof of the fact that (diamond, even
hole)-free graphs can be decomposed by clique cutsets into graphs of bounded
clique-width. In this paper we exhibit a class of (diamond, even hole)-free
graphs with no clique cutset that has unbounded rank-width (and hence
clique-width), so disproving Kloks' claim.

Another interesting subclass of even-hole-free graphs is the class of (cap,
even hole)-free graphs, where a \emph{cap} is a graph made of a hole and a
vertex that has exactly two neighbors on this hole, which are furthermore
adjacent. Cap-free graphs in general are decomposed by amalgams in
\cite{cckv-cap}. Recently, Conforti, Gerards and Pashkovich \cite{cgp}, show
how to obtain a polynomial-time algorithm for solving the maximum weighted
stable set problem on any class of graphs that is decomposable by amalgams
into basic graphs for which one can solve the maximum weighted stable set
problem in polynomial time. This leads to a polynomial-time algorithm for
solving the maximum weight stable set problem for (cap, even-hole)-free
graphs. Subsequently, Cameron, da Silva, Huang and Vu\v{s}kovi\'c \cite{cshv}
give an explicit construction of (cap, even hole)-free graphs, which is then
used to show that (triangle, even hole)-free graphs have tree-width at most 5,
and that (cap, even hole)-free graphs with no clique cutset have clique-width
at most 48 (and hence bounded rank-width). This implies that a number of
problems can be solved efficiently on this class, and in particular the class
can be colored in polynomial time.

\section{Preliminaries}
%        =============

Graphs are finite, simple and undirected unless stated otherwise. The vertex
set of a graph $G$ is denoted by $V(G)$ and the edge set by $E(G)$. A graph
$H$ is a \emph{subgraph} of a graph $G$, denoted by $H\subseteq G$, if
$V(H)\subseteq V(G)$ and $E(H)\subseteq E(G)$. For a graph $G$ and a subset
$X\subseteq V(G)$, we let $G[X]$ denote the subgraph of $G$ \emph{induced} by
$X$, i.e.\ $G[X]$ has vertex set $X$, and $E(G[X])$ consists of the edges of
$G$ that have both ends in $X$. A graph $H\subseteq G$ is an \emph{induced
  subgraph }of $G$, if $H=G[X]$ for some $X\subseteq V(G)$. Moreover, we let
$G\setminus X:= G[V(G)\setminus X]$. The set $X$ is a \emph{clique}, if $G[X]$
contains all possible edges. If $G$ is connected, $X$ is called a
\emph{clique cutset} if $X$ is a clique and $G\setminus X$ is disconnected.
%For a vertex $x\in V(G)$, we let $N_G(x)$ denote the set of neighbours of $x$ in $G$, together with $x$.

A \emph{tree} is a connected, acyclic graph. A \emph{leaf} of a tree is a node
incident to exactly one edge. For a tree $T$, we let $L(T)$ denote the set of
all leaves of $T$. A tree node that is not a leaf is called \emph{internal}.
A tree is \emph{cubic}, if it has at least two vertices and every internal
node has degree $3$. A \emph{path} is a tree where every node has degree at
most $2$. The (at most $2$) leaves of a path $P$ are called
\emph{end-vertices} of $P$. A \emph{$u,v$-path} is a path with end-vertices
$u$ and $v$. A graph $P$ is a \emph{subpath} of a graph $G$, if $P$ is a path
and $P\subseteq G$.

For a set $X$, let $2^X$ denote the set of all subsets of $X$. For sets $R$
and $C$, an \emph{$(R,C)$-matrix} is a matrix where the rows are indexed by
elements in $R$ and columns indexed by elements in $C$. For an $(R,C)$-matrix
$M$, if $X\subseteq R$ and $Y\subseteq C$, we let $M[X,Y]$ be the submatrix of
$M$ where the rows and the columns are indexed by $X$ and $Y$, respectively.
For a graph $G=(V,E)$, let $A_G$ denote the adjacency matrix of $G$ over the
binary field (i.e.\ $A_G$ is the $(V,V)$-matrix, where an entry is $1$, if and
only if, the column-vertex is adjacent to the row-vertex, and $0$ otherwise).
The \emph{cutrank function} of $G$ is the function
$\cutrk{G}\colon2^{V}\to\IN$, given by
\[\cutrk{G}(X)=\rank\left(A_G[X,V\sm X]\right),\]
where the rank is taken over the binary field.

A \emph{rank decomposition} of a graph $G$ is a pair $(T,\lambda)$, where $T$
is a cubic tree and $\lambda\colon V(G)\to L(T)$ is a bijection. If
$\left|V(G)\right| \le 1$, then $G$ has no rank decomposition. For every edge
$e\in E(T)$, the connected components of $T-e$ induce a partition $(A_e,B_e)$
of $L(T)$. The \emph{width} of an edge $e$ is defined as
$\cutrk{G}(\lambda^{-1}(A_e))$. The \emph{width} of $(T,\lambda)$, denoted by
$\width(T,\lambda)$, is the maximum width over all edges of $T$. The
\emph{rank-width} of $G$, denoted by $\rw(G)$, is the minimum integer $k$,
such that there is a rank decomposition of $G$ of width $k$. (If
$\left|V(G)\right| \le 1$, we let $\rw(G)=0$.)

\begin{remark}\label{rem:isubgraph-monotone}
  Let $G$ be a graph and $H\subseteq G$ be an induced subgraph of $G$.
  Then $\rw(H) \le \rw(G)$. %\qed
\end{remark} 

We say that a class $\cC$ of graphs has \emph{bounded} rank-width, if
there exists a constant $k\in \IN$, such that every $G\in \cC$
satisfies $\rw(G) \le k$. If such a constant does not exist, $\cC$ has
\emph{unbounded} rank-width.

We conclude the section with two lemmas that we will use in
Section~\ref{sec:lower-bound}.

\begin{lemma}\label{lem:rdec-path}
  Let $k\in \IN$. Let $G$ be a graph, $P\subseteq G$ an induced path,
  $(T,\lambda)$ a rank decomposition of $G$ of width at most $k$, and $e \in
  E(T)$. Let $(X,Y)$ be the bipartition of $V(P)$ induced by the two
  components of $T-e$. Then the induced graph $P[X]$ has at most $k+1$
  connected components.
\end{lemma}

\begin{proof} 
  Towards a contradiction, assume that $P[X]$ has at least $k+2$ components.
  Order the components (which are subpaths of $P$) according to their
  appearance along $P$. From each component, except for the first one, pick
  the first vertex. In this way we obtain a set $X'\subseteq X$ of at least
  $k+1$ vertices, each with one or two neighbours in $Y$ (two neighbours only
  if the component is a singleton vertex). Let $Y'$ be the set of vertices in
  $Y$ that are adjacent to a vertex in $X'$. Then each row of $A_P[X',Y']$ has
  one or two non-zero entries, and no two rows are equal. Ordering the
  vertices of $X'$ and $Y'$ according to their appearance on $P$ yields a
  matrix with blocks corresponding to subpaths of $P$, such that in each row
  the (at most two) non-zero entries appear consecutively. By the choice of
  $X'$, within each block there is at most one row with precisely one non-zero
  entry, while all other rows in that block have two non-zero entries. 
% consecuteve ones
  With this it is easy to see that the rows of each block are linearly
  independent, and it follows that $A_P[X',Y']$ has rank at least $k+1$. Since
  $P$ is induced, we have $A_P[X',Y']=A_G[X',Y']$, and hence the width of $e$
  is at least $k+1$, a contradiction to the width of $(T,\lambda)$ being at
  most $k$.
\end{proof}

We use the following definition, several variants of which exist in the
literature.

\begin{defeng}
  Let $T$ be a tree. We call an edge $e \in E(T)$ \emph{balanced}, if the
  partition $(A_e,B_e)$ of $L(T)$ satisfies $\frac13\left|L(T)\right| \le
  \left|A_e\right|$ and $\frac13\left|L(T)\right| \le \left|B_e\right|$.
\end{defeng}

The following lemma is well-known and we omit the proof.

\begin{lemma}\label{lem:balanced-edge}
  Every cubic tree has a balanced edge.
\end{lemma}

%\begin{proof} 
%Let $T$ be a cubic tree.
%For every edge $e \in E(T)$ and every subtree $T'$ obtained by removing $e$
%from $T$, we orient the edge $e$ in $T$ away from $T'$,
%if $\frac13\left|L(T)\right|> \left|L(T')\right|$. 
%If then there is a non-oriented edge $e$, 
%we are done. Assume now that all edges are oriented. Since $T$ is a tree,
%some node $s \in V(T)$ must be a sink. Note that $s$ cannot be a leaf. 
%But $T$ is cubic, and each of the three subtrees obtained from $T$ by
%deleting $s$ contains less than $\frac13$ of the vertices of $L(T)$, 
%a contradiction.
%\end{proof}

\begin{lemma}\label{lem:heavy-subpath}
  For $m,k\in \IN$ with $k>1$, let $G$ be a graph, $P\subseteq G$ be an
  induced path and $|V(G)|-|V(P)|=m$. Let $(T,\lambda)$ be a rank
  decomposition of $G$ of width at most $k$, and let $e\in E(T)$ be a balanced
  edge. Let $(X,Y)$ be the bipartition of $V(P)$ induced by $e$. Then each of
  the two induced subgraphs $P[X]$ and $P[Y]$ contains a connected component
  with at least $\left\lfloor\frac{|V(G)|-3m}{3(k+1)}\right\rfloor$ vertices.
\end{lemma}

\begin{proof} 
  Since $e$ is balanced, we have
  $\left|X\right|\geq\frac13\left|V(G)\right|-m$ and
  $\left|Y\right|\geq\frac13\left|V(G)\right|-m$. 
  By Lemma~\ref{lem:rdec-path}, both $P[X]$ and $P[Y]$ have at most $k+1$
  connected components, which proves the lemma.
\end{proof}

\section{Construction}
%        ============

In this section we construct a class of (diamond, even-hole)-free graphs
$(G_d)_{d \ge 1}$.

For $1 \le k \le d$, let
\[S_k=\{(a_1,a_2,\dots,a_{k-1},a_k)\,:\,a_1,a_2,\dots,a_{k-1}\in\{1,3\},\,a_k\in\{1,2,3,4\}\},\]
and $S^d=\bigcup_{k=1}^d S_k$. If $u\in S_k$, then we denote $l(u)=k$, and say
that the \emph{length of} $u$ is $k$.

In $S^d$, let $\preq$ denote the lexicographical order defined as follows.
%, which is defined in the following way: 
For $a=(a_1,a_2,\dots,a_k)\in S^d$ and
$b=(b_1,b_2,\dots,b_l)\in S^d$, $a\preq b$ if and only if $k \le l$ and
$a_i=b_i$ for $1\leq i \le k$, or $t=\min\{i\,:\,a_i \ne b_i\}$ is
well-defined and $a_t<b_t$. This order is a total order on the finite set
$S^d$, so we introduce the following notation:
\begin{itemize}
\item for $a\in S^d\sm\{(4)\}$, $s(a)$ is the smallest element (w.r.t.\ 
  $\preq$) of $S^d$ that is greater than $a$;
\item for $a\in S^d\sm\{(1)\}$, $p(a)$ is the greatest element (w.r.t.\ 
  $\preq$) of $S^d$ that is smaller than $a$.
\end{itemize}

Let $P'_d$ denote the path on vertex set $S^d$ connecting the vertices
according to the lexicographic order, and let $P_d$ be the path obtained from
$P'_d$ by subdividing every edge $uv\in E(P'_d)$ twice if $l(u) = l(v)$, and
once, otherwise. Finally, let $W_d=\{v_1,v_2,\dots,v_d\}$ be a set of (new)
vertices, such that $v_k$, for $1 \le k \le d$, is adjacent to all vertices of
$S_k$ and all other vertices of $W_d$. Then, $G_d$ is the graph induced by the
set $W_d\cup V(P_d)$. For vertices of $W_d$ we say that they are
\emph{centers} of $G_d$. Figure~\ref{fig:G4} shows $G_4$.

\begin{remark} \label{rem:size}
  For $d \ge 1$, the following hold:
  \begin{enumerate}[(i)]
  \item \label{s1} $|S^d|=\sum_{k=1}^d 4\cdot 2^{k-1}=4(2^d-1) \ge 2^{d+1}$,
    and
  \item\label{s2} $3|S^d|+d \ge |V(G_d)| \ge 2|S^d| \ge 2^{d+2}$. 
  \end{enumerate}
\end{remark}

\begin{proof}
  Part (\ref{s1}) follows from the fact that for $k=1$, the set $S_k$ contains
  $4$ vertices, and that the number of vertices in the set doubles whenever
  $k$ increases by one. Part (\ref{s2}) follows from Part (\ref{s1}) and the
  number of subdivision vertices added in the construction of $P_d$.
\end{proof}

\begin{remark} \label{rem:continuous}
  For $d \ge 1$, every $u\in S^d$, with $u \ne (4)$, satisfies
  $|l(u)-l(s(u))| \le 1$.
\end{remark}

\begin{figure}[htbp]
  \begin{center}  
    \newcommand{\NtoS}[1]{\ifcase #1 % number -> string
          (0)                      \or 0a    \or 0b    %   0,   1,   2,
      \or (1)                      \or 1a              %   3,   4,     
      \or    (1,1)                 \or 11a             %   5,   6,     
      \or         (1,1,1)          \or 111a            %   7,   8,     
      \or                (1,1,1,1) \or 1111a \or 1111b %   9,  10,  11,
      \or                (1,1,1,2) \or 1112a \or 1112b %  12,  13,  14,
      \or                (1,1,1,3) \or 1113a \or 1113b %  15,  16,  17,
      \or                (1,1,1,4) \or 1114a           %  18,  19,     
      \or         (1,1,2)          \or 112a  \or 112b  %  20,  21,  22,
      \or         (1,1,3)          \or 113a            %  23,  24,     
      \or                (1,1,3,1) \or 1131a \or 1131b %  25,  26,  27,
      \or                (1,1,3,2) \or 1132a \or 1132b %  28,  29,  30,
      \or                (1,1,3,3) \or 1133a \or 1133b %  31,  32,  33,
      \or                (1,1,3,4) \or 1134a           %  34,  35,     
      \or         (1,1,4)          \or 114a            %  36,  37,     
      \or    (1,2)                 \or 12a   \or 12b   %  38,  39,  40,
      \or    (1,3)                 \or 13a             %  41,  42,     
      \or         (1,3,1)          \or 131a            %  43,  44,     
      \or                (1,3,1,1) \or 1311a \or 1311b %  45,  46,  47,
      \or                (1,3,1,2) \or 1312a \or 1312b %  48,  49,  50,
      \or                (1,3,1,3) \or 1313a \or 1313b %  51,  52,  53,
      \or                (1,3,1,4) \or 1314a           %  54,  55,     
      \or         (1,3,2)          \or 132a  \or 132b  %  56,  57,  58,
      \or         (1,3,3)          \or 133a            %  59,  60,     
      \or                (1,3,3,1) \or 1331a \or 1331b %  61,  62,  63,
      \or                (1,3,3,2) \or 1332a \or 1332b %  64,  65,  66,
      \or                (1,3,3,3) \or 1333a \or 1333b %  67,  68,  69,
      \or                (1,3,3,4) \or 1334a           %  70,  71,     
      \or         (1,3,4)          \or 134a            %  72,  73,     
      \or    (1,4)                 \or 14a             %  74,  75,     
      \or (2)                      \or 2a    \or 2b    %  76,  77,  78,
      \or (3)                      \or 3a              %  79,  80,     
      \or    (3,1)                 \or 31a             %  81,  82,     
      \or         (3,1,1)          \or 311a            %  83,  84,     
      \or                (3,1,1,1) \or 3111a \or 3111b %  85,  86,  87,
      \or                (3,1,1,2) \or 3112a \or 3112b %  88,  89,  90,
      \or                (3,1,1,3) \or 3113a \or 3113b %  91,  92,  93,
      \or                (3,1,1,4) \or 3114a           %  94,  95,     
      \or         (3,1,2)          \or 312a  \or 312b  %  96,  97,  98,
      \or         (3,1,3)          \or 313a            %  99, 100,     
      \or                (3,1,3,1) \or 3131a \or 3131b % 101, 102, 103,
      \or                (3,1,3,2) \or 3132a \or 3132b % 104, 105, 106,
      \or                (3,1,3,3) \or 3133a \or 3133b % 107, 108, 109,
      \or                (3,1,3,4) \or 3134a           % 110, 111,     
      \or         (3,1,4)          \or 314a            % 112, 113,     
      \or    (3,2)                 \or 32a   \or 32b   % 114, 115, 116,
      \or    (3,3)                 \or 33a             % 117, 118,     
      \or         (3,3,1)          \or 331a            % 119, 120,     
      \or                (3,3,1,1) \or 3311a \or 3311b % 121, 122, 123,
      \or                (3,3,1,2) \or 3312a \or 3312b % 124, 125, 126,
      \or                (3,3,1,3) \or 3313a \or 3313b % 127, 128, 129,
      \or                (3,3,1,4) \or 3314a           % 130, 131,     
      \or         (3,3,2)          \or 332a  \or 332b  % 132, 133, 134,
      \or         (3,3,3)          \or 333a            % 135, 136,     
      \or                (3,3,3,1) \or 3331a \or 3331b % 137, 138, 139,
      \or                (3,3,3,2) \or 3332a \or 3332b % 140, 141, 142,
      \or                (3,3,3,3) \or 3333a \or 3333b % 143, 144, 145,
      \or                (3,3,3,4) \or 3334a           % 146, 147,     
      \or        (3,3,4)           \or 334a            % 148, 149,     
      \or   (3,4)                  \or 34a             % 150, 151,     
      \or (4)                      \or 4a    \or 4b    % 152, 153, 154,
      \or (5)                                          % 155
      \fi}

    \newcommand{\layer}[1]{\ifcase #1 0,155%                    layer 0
      \or 3,76,79,152%                                          layer 1
      \or 5,38,41,74,81,114,117,150%                            layer 2
      \or 7,20,23,36,43,56,59,72,83,96,99,112,119,132,135,148%  layer 3
      \or 9,12,15,18,25,28,31,34,45,48,51,54,61,64,67,70,85,88,91,94,%
          101,104,107,110,121,124,127,130,137,140,143,146%      layer 4
      \or 4,6,8,10,11,13,14,16,17,19,21,22,24,26,27,29,30,32,33,35,37,39,%
          40,42,44,46,47,49,50,52,53,55,57,58,60,62,63,65,66,68,69,71,73,75,%
          77,78,80,82,84,86,87,89,90,92,93,95,97,98,100,102,103,105,106,108,%
          109,111,113,115,116,118,120,122,123,125,126,128,129,131,133,134,%
          136,138,139,141,142,144,145,147,149,151%              layer 5
      \fi}

    \newcommand{\col}[1]{\ifcase #1 %
      black\or blue\or red\or green!70!black\or yellow!60!red\else none\fi}

    \tikzset{v/.style={circle, draw=black, minimum size=1.5mm, inner sep=0pt},
             n/.style={draw=none}
            }

    \begin{tikzpicture}
      \node[v, fill=\col{1}, label=180:$v_1$] (x1) at (195:2.0) {};
      \node[v, fill=\col{2}, label=135:$v_2$] (x2) at (135:2.5) {};
      \node[v, fill=\col{3}, label=105:$v_3$] (x3) at ( 65:2.0) {};
      \node[v, fill=\col{4}, label=  0:$v_4$] (x4) at (320:2.0) {};
      \edef\nlist{\layer{5}}
      \foreach \num in \nlist{
        \pgfmathsetmacro{\angle}{270-360/155*\num}
        \node[v] (\num) at (\angle : 6.0) {};
      }
      \foreach \ind in {4,3,2,1}{ 
        \edef\nlist{\layer{\ind}}
        \pgfmathsetmacro{\dista}{6.3+\ind/7}
        \foreach \num in \nlist{
          \pgfmathsetmacro{\angle}{270-360/155*\num}
          \node[v, fill=\col{\ind}] (\num) at (\angle:6.0) {};
          \draw[\col{\ind}] (x\ind)--(\num);
          \node[n,rotate=\angle] () at (\angle:\dista) {\NtoS{\num}};
        }
      }
      \foreach[count=\mum from 3] \num in {4,5,...,152} \draw (\num)--(\mum);
      \draw (x1)--(x2)--(x3)--(x4)--(x1)--(x3)  (x2)--(x4);
      \node[v, fill=\col{4}, label=  0:$v_4$] (x4) at (320:2.0) {};
    \end{tikzpicture}
  \end{center}
  \caption{The graph $G_4$}
  \label{fig:G4}
\end{figure}

Let us introduce some additional notation for the elements of $S^d$. For
$a,b\in S^d$, \emph{interval} $[a,b]$ is the set $\{c\in S^d\,:\, a\preq
c\preq b\}$. We say that an interval $[a,b]$ is \emph{proper} if for all
$c\in [a,b]\sm\{a,b\}$, $l(c)\not\in\{l(a),l(b)\}$. Note that
$[a,b]=\bigcup_{a\preq c\prec b}[c,s(c)]$. For an interval $[c,s(c)]$, $a\preq
c\prec b$, we say that it is a \emph{step of $[a,b]$}, and if additionally
$l(c)=l(s(c))$, we say that this step is \emph{flat}.

\begin{lemma}\label{OddEqual}
  Let $a,b\in S^d$. If $[a,b]$ is a proper interval such that $l(a)=l(b)$,
  then it contains an odd number of flat steps.
\end{lemma}

\begin{proof}
  Our proof is by induction on the number of elements of $[a,b]$. If $[a,b]$
  has only 2 elements, that is if $b=s(a)$, then the lemma trivially holds.
  Let $a=(a_1,a_2,\dots,a_k)$.

  \begin{enumerate}[{Case }1.]
% \medskip\noindent\emph{Case 1.} $a_k=2$.
  \item  $a_k=2$.
% \noindent

    In this case $b=(a_1,a_2,\dots,a_{k-1},3)$ and $[a,b]=\{a,b\}$, so the
    conclusion trivially follows.

% \medskip\noindent\emph{Case 2.} $a_k\in\{1,3\}$.
  \item $a_k\in\{1,3\}$.
% \noindent

    In this case $b=(a_1,a_2,\dots,a_{k-1},a_k+1)$. If $k=d$, then
    $[a,b]=\{a,b\}$, and the conclusion follows. So, let $k<d$. Then
    \[[a,b]=[a,a^{(1)}]\cup [a^{(1)},a^{(2)}]\cup [a^{(2)},a^{(3)}]\cup
    [a^{(3)},a^{(4)}]\cup [a^{(4)},b],\] 
    where $a^{(i)}=(a_1,a_2,\dots,a_k,i)$, for $1 \le i \le 4$.   
    Since $s(a)=a^{(1)}$ and $s(a^{(4)})=b$, the number of flat steps of
    $[a,b]$ is the sum of the numbers of flat steps of $[a^{(i)},a^{(i+1)}]$,
    for $1 \le i \le 3$. Note that $a^{(i)}$ and $a^{(i+1)}$, for $1 \le i
    \le 3$, are consecutive $(k+1)$-tuples of $S^d$, i.e.\ the interval
    $[a^{(i)},a^{(i+1)}]$ is proper. Therefore, by induction, each of the
    intervals $[a^{(i)},a^{(i+1)}]$, for $1 \le i \le 3$, has an odd number of
    flat steps, and hence so does the interval $[a,b]$.

% \medskip\noindent\emph{Case 3.} $a_k=4$.
  \item $a_k=4$.
% \noindent

    In this case $a_{k-1}\in\{1,3\}$, so 
    \[a=(a_1,\dots,a_{i-1},1,\underbrace{3,\dots,3}_{k-i-1},4),\] 
    where $1 \le i \le k-1$ ($a$ has at least one coordinate equal to 1, since
    there does not exist a $k$-tuple in $S^d$ which is larger than the $k$-tuple
    $({3,\dots,3},4)$).
  
    If $i=k-1$, then $s(a) = (a_1,\dots,a_{i-1},2)$,
    $s(s(a)) = (a_1,\dots,a_{i-1},3)$ and
    $s(s(s(a))) = (a_1,\dots,a_{i-1},3,1)=b$,
    and hence the interval $[a,b]$ has one flat step.

    So, let $i<k-1$. Then
    \begin{align*}
      s(a)&=(a_1,\dots,a_{i-1},1,\underbrace{3,\dots,3}_{k-i-2},4),\\
      p(b)&=(a_1,\dots,a_{i-1},3,\underbrace{1,\dots,1}_{k-i-1}),\\
      b&=(a_1,\dots,a_{i-1},3,\underbrace{1,\dots,1}_{k-i}).
    \end{align*}
    So, the number of flat steps of the interval $[a,b]$ is the same as the
    number of flat steps of the interval $[s(a),p(b)]$. Since $s(a)$ and $p(b)$
    are consecutive $(k-1)$-tuples of $S^d$, the interval $[s(a),p(b)]$ is
    proper, and the conclusion follows by induction. \qedhere
  \end{enumerate}
\end{proof}

\begin{lemma}\label{ZeroEqual}
  Let $a,b\in S^d$. If $[a,b]$ is a proper interval such that $l(a) \ne l(b)$,
  then it does not contain a flat step.
\end{lemma}

\begin{proof}
  Note that the set $S^d$ is symmetric, so we may assume that $l(a)>l(b)$.  
  Let $a=(a_1,\dots,a_{k-1},a_k)$.  
  If $a_k<4$, then $a \prec (a_1,\dots,a_{k-1},a_k+1)$, and hence $[a,b]$ is
  not proper, since there does not exist $c\in S^d$ such that
  $(a_1,\dots,a_{k-1},a_k)\prec c\prec (a_1,\dots,a_{k-1},a_k+1)$ and $l(c)<k$. 
  So, $a_k=4$. If $a=(\underbrace{3,\dots,3}_{k-1},4)$, then
  $b=(\underbrace{3,\dots,3}_{l-1},4)$, where $l=l(b)$, and the conclusion
  follows. So, let
  \[a=(a_1,\dots,a_{i-1},1,\underbrace{3,\dots,3}_{k-i-1},4)\,,\]   
  where $1 \le i \le k-1$, and
  \[a'=(a_1,\dots,a_{i-1},3,\underbrace{1,\dots,1}_{k-i})\,.\] 
  The elements of the interval $[a,a']$ are the following (given in increasing
  order):
  \begin{align*}
    &(a_1,\dots,a_{i-1},1,\underbrace{3,\dots,3}_{k-i-1},4),(a_1,\dots,a_{i-1},1,\underbrace{3,\dots,3}_{k-i-2},4),\dots,(a_1,\dots,a_{i-1},2),\\
    &(a_1,\dots,a_{i-1},3),(a_1,\dots,a_{i-1},3,1),\dots,(a_1,\dots,a_{i-1},3,\underbrace{1,\dots,1}_{k-i}).
  \end{align*}
  Since $[a,b]$ is proper, it holds $b \prec a'$. Additionally, since $[a,b]$
  does not contain an element of length equal to $l(b)$, $b$ is an element of
  $[a,a']$ from the first row of the given list. Now it is clear that $[a,b]$
  contains zero flat steps.
\end{proof}

For an interval $[a,b]$ in $S^d$, let $P_{[a,b]}$ be the path of $G_d$ induced
by $\bigcup_{a\preq c\prec b} V(P_{c})$. Since path $P_c$ is of odd length if
and only if $l(c)=l(s(c))$, path $P_{[a,b]}$ is of odd length if and only if
$[a,b]$ contains an odd number of flat steps.

\begin{theorem}\label{theo:gd-ehf}
  The graph $G_d$ is (diamond, even hole)-free for all $d \ge 1$ and
  $G_d$ has no clique cutset for all $d \ge 2$.
\end{theorem}

\begin{proof}
  First, suppose that $G_d$ contains a diamond $D$ for some $d \ge 1$. Since
  $P_d$ is a path, $V(D)\not\subseteq V(P_d)$, and since $D$ is not a clique
  $V(D)\not\subseteq W_d$. The neighborhood in $P_d$ of every vertex of $W_d$
  is a stable set, so $|V(D)\cap V(P_d)| \le 2$. On the other hand, every
  vertex of $P_d$ is adjacent to at most one vertex of $W_d$, so $|V(D)\cap
  V(W_d)| \le 2$. Hence, $|V(D)\cap V(P_d)|=|V(D)\cap W_d|=2$. But then $D$
  has at most 4 edges, a contradiction.
  
  Now, suppose that $G_d$ contains an even hole $H$ for some $d \ge 1$. Since
  $P_d$ is a path, $V(H)\cap W_d \ne \es$, and since $W_d$ is a clique
  $|V(H)\cap V(W_d)| \le 2$. First suppose that $V(H)\cap V(W_d)=\{v_k\}$,
  for some $1 \le k \le d$. Since $v_k$ has exactly two neighbors in $H$,
  $V(H)=\{v_k\}\cup V(P_{[a,b]})$, where $a,b\in S^d$ are such that
  $l(a)=l(b)=k$ and the interval $[a,b]$ is proper. Then, by Lemma
  \ref{OddEqual}, interval $[a,b]$ contains an odd number of flat steps, and
  hence path $P_{[a,b]}$ and hole $H$ are of odd length, a contradiction. So,
  $V(H)\cap V(W_d)=\{v_k,v_l\}$, for some $1 \le k<l \le d$. Then
  $V(H)=\{v_k,v_l\}\cup V(P_{[a,b]})$, where $a,b\in S^d$ are such that
  $\{l(a),l(b)\}=\{k,l\}$ and the interval $[a,b]$ is proper. Then, by Lemma
  \ref{ZeroEqual}, interval $[a,b]$ does not contain a flat step, and hence
  path $P_{[a,b]}$ is of even length, i.e.\ the hole $H$ is of odd length
  (since the length of $H$ is by 3 larger than the length of $P_{[a,b]}$), a
  contradiction.

%Therefore, the class $(G_d)_{d \ge 1}$ is (diamond, even hole)-free. 
  
  Let $d \ge 2$ and suppose that $G_d$ has a clique cutset $K$. We distinguish
  between three cases. First, if $K \subseteq W_d$ then $K$ does not separate
  since $P_d$ is a path and every vertex in $W_d \sm K$ has a neighbor in
  $P_d$. Second, if $K \subseteq V(P_d)$ then $P_d - K$ has two components.
  In $G_d - K$ these are connected via $W_d$ since $d \ge 2$. Hence we are in
  the third case and may assume $K \cap W_d \ne \es$ and $K \cap V(P_d) \ne
  \es$. By construction, no vertex of $P_d$ is contained in a triangle, and
  hence $|K| \le 2$. Consequently $K=\{u,v_i\}$ for $u \in V(P_d)$ and $1 \le
  i \le d$. The vertex $u$ is neither (1) nor (4) since both are adjacent to
  $v_1 \in W_d$ and neither $\{(1),v_1\}$ nor $\{(4),v_1\}$ are cutsets of
  $G_d$. It follows that (1) and (4) are separated by $K$. Since $v_1$ is
  adjacent to both (1) and (4) we have $i=1$, and hence $u$ is (2) or (3).
  But then $v_2$ has a neighbor in both components of $P_d - u$, a
  contradiction.
\end{proof}

\section{Lower bound}\label{sec:lower-bound}
%        ===========

In this section we prove that the rank-width of the class $(G_d)_{d \ge 1}$
constructed in the previous section is unbounded.

\begin{lemma} \label{lem:suffix}
  If $d \ge 1$ and $P$ is a subpath of $P_d$ such that $|V(P)\cap S_i| \ge 3$ 
  for some $i$ $(1 \le i \le d)$,
  then $V(P)\cap S_j \ne \es$ for every $j$ satisfying $i \le j \le d$.
\end{lemma}

\begin{proof}
  Since $|V(P)\cap S_i| \ge 3$, there exist two vertices of the form
  $(a_1,\dots,a_{i-1},1)$ and $(a_1,\dots,a_{i-1},2)$, or two vertices of the
  form $(a_1,\dots,a_{i-1},3)$ and $(a_1,\dots,a_{i-1},4)$ in $P$, where
  $a_k\in \{1,3\}$ for $1 \le k < i$. But then, by the definition of the order
  $\preq$ for $S^d$, $P$ must contain some vertex of length $j$ for every $j$
  satisfying $i \le j \le d$.
\end{proof}

\begin{lemma} \label{lem:large-color}
  If $P$ is a subpath of $P_d$ such that $|V(P)| \ge c |V(G_d)|$, where
  $0<c<1$ and $d>2\lfloor\log_2{\frac{1}{c}}\rfloor+4$,
  then $V(P)\cap S_j \ne \es$ for every $j$ satisfying 
  $\lfloor\log_2{\frac{1}{c}}\rfloor+3 \le j\le d$.
\end{lemma}

\begin{proof}
  If $V(P)\cap S_j \ne \es$ for every $j\in\{1,\dots,d\}$, then the
  conclusion trivially holds. Hence, we may assume that $V(P)\cap S_j = \es$
  for some $j \in \{1,\dots,d\}$.

% \vspace{2ex}

  \begin{enumerate}[{Claim }1.]
% \noindent\emph{Claim 1.} $|V(P)|>6d$.
  \item\label{claim1} $|V(P)|>6d$.
    
    \emph{Proof of Claim \ref{claim1}:} Suppose that $|V(P)| \le 6d$.   
    Since $|V(P)| \ge c |V(G_d)| \ge c\cdot 2^{d+2}$ (the first inequality is
    by the assumption, and the second by Remark \ref{rem:size}), it follows
    that $6d \ge c\cdot 2^{d+2}$, which is equivalent to $\log_2{\frac{1}{c}}
    \ge d - \log_2{d} +2-\log_2{6}$. Since $d-\log_2{d} \ge \frac{d}{2}$, for
    all $d \ge 4$ (which is the case by assumption), and $2-\log_2{6}>-1$, we
    have that $\log_2{\frac{1}{c}}>\frac{d}{2}-1$, which is equivalent to
    $d<2\log_2{\frac{1}{c}}+2$, a contradiction. This completes the proof of
    Claim \ref{claim1}.

% \vspace{2ex}

% \noindent \emph{Claim 2.} For some $t\in\{1,\dots,d\}$, $|V(P)\cap S_t|\ge3$.
  \item\label{claim2}  For some $t \in \{1,\dots,d\}$, $|V(P)\cap S_t| \ge 3$.
    
    \emph{Proof of Claim \ref{claim2}:} 
    Suppose that for all $t \in \{1,\dots,d\}$, $|V(P)\cap S_t| \le 2$.
    Let $a'$ and $b'$ be the endnodes of $P$, and let $a$ (resp.~$b$) be the
    first (resp.\ last) vertex of $S^d$ encountered when traversing $P$ from
    $a'$ to $b'$. Since for some $j$, $V(P)\cap S_j=\es$, the interval
    $[a,b]$ contains at most $d-2+1+d-2=2d-3$ steps (note that this bound can
    be achieved when $[a,b]$ contains vertices $(2)$ and $(3)$, the $d-2$
    elements of $S^d$ that precede $(2)$, and the $d-2$ elements of $S^d$ that
    succeed $(3)$). For each step $[u,s(u)]$, the $u,s(u)$-subpath of $P$ is
    of length at most three. The $a,a'$-subpath of $P$ and the $b,b'$-subpath
    of $P$ are each of length at most two. It follows that the length of $P$
    is at most $3(2d-3)+2 \cdot 2 = 6d-5$, and hence $|V(P)| \le 6d$,
    contradicting Claim \ref{claim1}. 
    This completes the proof of Claim \ref{claim2}.
 
% \vspace{2ex}
 
    By Claim \ref{claim2} and Lemma \ref{lem:suffix}, for some $i<d$,
    $V(P)\cap S_i=\es$ and $V(P)\cap S_j \ne \es$ for $j \in \{i+1,\dots,d\}$. 
    By Remark~\ref{rem:continuous}, $V(P)\cap S_j=\es$ for $j\in \{1,\dots,i\}$.
    Therefore, there exist two vertices $u,v \in S_i$, $u \preq v$, such that
    $P$ is contained in the subpath $P'$ of $P_d$ from $u$ to $v$ and 
    $V(P') \cap S_i = \{u,v\}$. Let $u=(a_1,\dots,a_i)$.
 
%  \vspace{2ex}
 
% \noindent\emph{Claim 3. $a_i\in \{ 1,3\}$.}
  \item\label{claim3} $a_i\in \{ 1,3\}$.

  \emph{Proof of Claim \ref{claim3}:}
  We consider the following cases:
  \begin{itemize}
  \item If $a_i=2$ then $v=s(u)$. Hence, $|V(P')|= 4$.
  \item If $a_i=4$, then $u=(a_1,\dots,a_{i'-1},1,3,\dots,3,4)$, where     
    $1 \le i' \le i-1$ ($u$ has at least one coordinate equal to $1$,
    otherwise there does not exist a tuple in $S_i$ which is larger than $u$).
    Since $v$ is the next element in $S_i$ which is larger than $u$,
    $v=(a_1,\dots,a_{i'-1},3,1,\dots,1)$. By the discussion in the proof of
    Lemma \ref{ZeroEqual}, the number of elements of $S^d$ in the interval
    $[u,v]$ is $2(i-i'+1)$ and we have that $2(i-i'+1) \le 2i \le 2d$. Since
    there are at most two vertices of $P'$ between any two consecutive
    elements in $S_d$, $|V(P')| \le 3\cdot 2d = 6d$.
  \end{itemize} 
  Both cases contradict Claim \ref{claim1}.
  This completes the proof of Claim \ref{claim3}.
 \end{enumerate}
% \vspace{2ex}
 
 Since there are at most two vertices of $P'$ between any two consecutive
 elements in $S_d$ and by Claim~\ref{claim3}, $|V(P')| \le 3|[u,v]| =
 3(\sum_{j=0}^{d-i-1} 4\cdot 2^j+2) < 12(\sum_{j=0}^{d-i-1}2^j + 1) = 
 12\cdot 2^{d-i} < 2^{d-i+4}$. So by Remark~\ref{rem:size}, we have that
 \[ 2^{d-i+4} > |V(P')| \ge |V(P)| \ge c |V(G_d)|  \ge c\cdot 2^{d+2}\,. \]
 Hence $2^{2-i}>c$, or equivalently $i<2+\log_2{\frac{1}{c}}$,
 proving the lemma.
\end{proof}

\begin{lemma}\label{lem:rw-unbounded}
  For any $d \ge 22$ we have $\rw(G_d)> d/3$.
\end{lemma}

\begin{proof}
  Suppose that $\rw(G_d) \le k = d/3$. Let $(T,\lambda)$ be a rank
  decomposition of $G_d$ of width at most $k$. Let $e\in E(T)$ be a balanced
  edge (it exists by Lemma \ref{lem:balanced-edge}), and let $M$ be the
  adjacency matrix of $G_d$. Let $(X,Y)$ be the bipartition of $V(G_d)$
  induced by $e$. Appying Lemma \ref{lem:heavy-subpath} for $G_d$ and the path
  $P_d$ ($|V(G_d)|-|V(P_d)|=d$), there exist two subpaths $P_X$, $P_Y$ of
  $P_d$ in $G_d[X]$ and $G_d[Y]$, respectively, such that $|V(P_X)|,|V(P_Y)|
  \ge \left\lfloor\frac{|V(G_d)|-3d}{3(k+1)}\right\rfloor \ge
  \frac{|V(G_d)|}{4(k+1)}$ (note that the second inequality holds by Remark
  \ref{rem:size} and the fact that $d \ge 22$). Applying Lemma
  \ref{lem:large-color} (using the fact that $d \ge 22$) with
  $c=\frac{1}{4(k+1)}$ and letting $c'=\lfloor\log_2(\frac{1}{c})\rfloor +3 =
  \lfloor \log_2(k+1) \rfloor +5$, we have $V(P_X)\cap S_j \ne \es$ and
  $V(P_Y) \cap S_j \ne \es$ for every $j$ satisfying $c' \le j \le d$.
  W.l.o.g.\ let $X$ be the set containing at least half of the center vertices
  in $\{v_{c'},\dots,v_d\}$. Let $I = \{i \in \{c',\dots,d\} \mid v_i\in X\}$
  (the set of indices of center vertices in $X$), and fix a vertex
  $a_i \in Y \cap S_i$ for every $i \in I$, which exists because
  $V(P_Y)\cap S_i \ne \es$.
% by remark \ref{rem:long-enough-paths}.
  We have $|I| \ge \frac{d-c'+1}{2}$. Let $S_X=\{v_i \mid i\in I\}$ and 
  $S_Y=\{a_i \mid i\in I\}$.
%, where $a_i$ is some vertex in $Y\cap S_i$ ($a_i$ exists because $P_Y\cap S_i \ne \es$).
  Note that $S_X\subseteq X$ and $S_Y\subseteq Y$. Because each vertex
  $v_i$ in $S_X$ has exactly one neighbor in $S_Y$ (namely $a_i$),
  we have that $M[S_X,S_Y]=\mathbf{1}_{|I|}$ (identity matrix). Therefore,
  $\rank(M[S_X,S_Y])=|I|$. We have 
  \[ k \ge \width(T,\lambda) \ge \cutrk{G}(X) = \rank(M[X,Y])
       \ge \rank(M[S_X,S_Y]) =   |I| \ge \frac{d-c'+1}{2}\,,\]
  which is equivalent to
  $d \le 2k+c'-1 = 2k + \lfloor\log_2(k+1)\rfloor + 4
                 = 2d/3 + \lfloor\log_2(d/3+1)\rfloor + 4$,
   a contradiction since $d \ge 22$.
\end{proof}

From Lemma~\ref{lem:rw-unbounded} and Remark~\ref{rem:size} we obtain that 
the rankwidth of $G_d$ grows at least logarithmically with $|V(G_d)|$, 
since if $d \ge 22$ then $\rw(G_d ) > d/3 \ge (\log_2 |V(G_d)| - 4)/3$. 
From Theorem~\ref{theo:gd-ehf} and Lemma~\ref{lem:rw-unbounded} we have the
following theorem.

% \begin{corollary}\label{cor:rw-unbounded}
%   If $d \ge 22$, then $\rw(G_d )  \ge d/3  \ge (\log_2 |V(G_d)| - 4)/3$.
% \end{corollary}

% From Theorem~\ref{theo:gd-ehf} and Corollary~\ref{cor:rw-unbounded}
% we obtain the following theorem.

\begin{theorem}
  The family of (diamond, even hole)-free graphs $G_d$, $d \ge 2$, without
  clique cutsets has unbounded rank-width.
\end{theorem}

For completeness, observe that $\rw(G_d) \le d+1$ for all $d\in\IN$.
To see this, take a cubic tree $T$ with $\left| V(G_d)\right|$ leaves,
where the internal nodes form a path. Via the bijection
$\lambda: V(G_d) \to L(T)$, pick the linear ordering on 
$W_d \cup V(P_d)$, which starts with $v_1, v_2, v_3, \dots, v_d$, followed
by the vertices of $P_d$ in their canonical order (see Figure
\ref{fig:dec}).

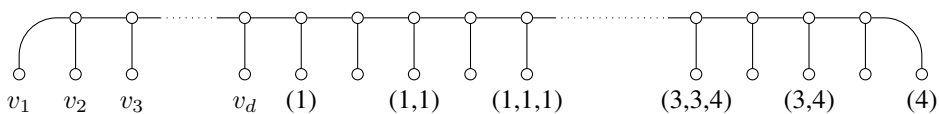
\begin{figure}[htbp]
  \begin{center}
    \newcommand{\vv}[1]{{\phantom{(}$v_{#1}$\phantom{)}}}
    \tikzset{v/.style={circle, draw=black, minimum size=1.5mm, inner sep=0pt}}
    \begin{tikzpicture}[scale=0.75]
      \node[v, label=below:\vv{1}] (v1) at (0,0) {};
      \node[v, label=below:\vv{2}] (v2) at (1,0) {};
      \node[v] (pv2) at (1,1) {}; \draw (v2)--(pv2);
      \draw[rounded corners=5mm] (v1)--(0,1)--(pv2);
      \node[v, label=below:\vv{3}] (v3) at (2,0) {};
      \node[v] (pv3) at (2,1) {}; \draw (v3)--(pv3)--(pv2)  (pv3)--(2.5,1);
      \node[v, label=below:\vv{d}] (vd) at (4,0) {};
      \node[v] (pvd) at (4,1) {}; \draw (vd)--(pvd)--(3.5,1);
      \draw[dotted] (2.5,1)--(3.5,1);
      \node[v, label=below:{(1)}] (1) at (5,0) {};
      \node[v] (p1) at (5,1) {}; \draw (1)--(p1)--(pvd);
      \node[v] (1a) at (6,0) {};
      \node[v] (p1a) at (6,1) {}; \draw (1a)--(p1a)--(p1);
      \node[v, label=below:{(1,1)}] (11) at (7,0) {};
      \node[v] (p11) at (7,1) {}; \draw (11)--(p11)--(p1a);
      \node[v] (11a) at (8,0) {};
      \node[v] (p11a) at (8,1) {}; \draw (11a)--(p11a)--(p11);
      \node[v, label=below:{(1,1,1)}] (111) at (9,0) {};
      \node[v] (p111) at (9,1) {}; \draw (111)--(p111)--(p11a)  (p111)--(9.5,1);
      \node[v, label=below:{(3,3,4)}] (334) at (12,0) {};
      \node[v] (p334) at (12,1) {}; \draw (334)--(p334)--(11.5,1);
      \draw[dotted] (9.5,1)--(11.5,1);
      \node[v] (334a) at (13,0) {};
      \node[v] (p334a) at (13,1) {}; \draw (334a)--(p334a)--(p334);
      \node[v, label=below:{(3,4)}] (34) at (14,0) {};
      \node[v] (p34) at (14,1) {}; \draw (34)--(p34)--(p334a);
      \node[v] (34a) at (15,0) {};
      \node[v] (p34a) at (15,1) {}; \draw (34a)--(p34a)--(p34);
      \node[v, label=below:{(4)}] (4) at (16,0) {};
      \draw[rounded corners=5mm] (4)--(16,1)--(p34a);
    \end{tikzpicture}
    \caption{A rank decomposition of $G_d$ of width at most $d+1$.}
    \label{fig:dec}
  \end{center}
\end{figure}

Let $e$ be an edge of $T$ and let $(X,Y)$ be the bipartition of $V(G_d)$
induced by $e$. Since $\rank(M[X,Y]) \le \min(|X|,|Y|)$ we may assume
$|X|,|Y| > d$ and $\{v_1,v_2,\dots,v_d\} \subseteq X$. Now the vertices in $Y$
have at most $d+1$ different neighbours in $X$. Hence the width of $e$ is at
most $d+1$, proving that $\rw(G_d) \le d+1$.

\bibliographystyle{amsplain}
\bibliography{rw-ehf}

\end{document}